\documentclass[runningheads]{lipics-v2021}

\date{}
\usepackage{multicol}
\usepackage{latexsym,amsmath,amssymb,stmaryrd,mathtools}
\usepackage{graphicx,epsfig,tikz,color,xcolor,bbding,algorithm2e}
\usepackage{float}
\usepackage{caption}
\usepackage{subcaption}
\usepackage{comment}
\usetikzlibrary{shapes.geometric,positioning,chains,fit,calc}

\newcommand{\Oh}{{\cal O}}

\newcommand{\FST}{{\text{FeST}}}

\newcommand{\LCP}{{\text{LCP}}}
\newcommand{\lex}{{\textrm{lex}}}

\def\pos{\textit{pos}}
\def\node{\textit{node}}
\def\root{\textit{root}}
\def\splay{\textit{splay}}
\def\geomsum{\textit{geomsum}}
\renewcommand{\epsilon}{\varepsilon}

\newcommand{\makestring}{\texttt{make-string}}
\newcommand{\access}{\texttt{access}}
\newcommand{\retrieve}{\texttt{retrieve}}
\newcommand{\ins}{\texttt{insert}}
\newcommand{\del}{\texttt{delete}}
\newcommand{\repl}{\texttt{substitute}}

\newcommand{\introduce}{\texttt{introduce}}
\newcommand{\extract}{\texttt{extract}}
\newcommand{\equal}{\texttt{equal}}
\newcommand{\lcp}{\texttt{lcp}}
\newcommand{\reverse}{\texttt{reverse}}
\newcommand{\map}{\texttt{map}}
\newcommand{\isolate}{\texttt{isolate}}
\newcommand{\find}{\texttt{find}}
\newcommand{\fix}{\texttt{fix}}
\newcommand{\fixm}{\texttt{fixm}}
\newcommand{\ttrotate}{\texttt{rotate}}

\newcommand{\leftc}{.\mathrm{left}}
\newcommand{\rightc}{.\mathrm{right}}
\newcommand{\car}{.\mathrm{char}}
\newcommand{\size}{.\mathrm{size}}
\newcommand{\fp}{.\mathrm{fp}}
\newcommand{\power}{.\mathrm{power}}

\newcommand{\rev}{.\mathrm{rev}}
\newcommand{\mapf}{.\mathrm{map}}
\newcommand{\fprev}{.\mathrm{fprev}}
\newcommand{\mfp}{.\mathrm{mfp}}
\newcommand{\mfprev}{.\mathrm{mfprev}}
\newcommand{\rc}{.\mathrm{rc}}

\newcommand{\Fest}{\textrm{FeST}}

\newcommand{\start}{\mathrm{start}}

\author{Zsuzsanna Lipt\'ak}{Dipartimento di Informatica, University of Verona,
Italy}{zsuzsanna.liptak@univr.it}{https://orcid.org/0000-0002-3233-0691}{Partially funded by the MUR PRIN project Nr.\ 2022YRB97K 'PINC' (Pangenome INformatiCs. From Theory to Applications) and
by the INdAM-GNCS Project CUP$\_$E53C23001670001 (Compressione, indicizzazione, analisi e confronto di dati biologici).}
\author{Francesco Masillo}{Dipartimento di Informatica, University of Verona,
Italy}{francesco.masillo@univr.it}{https://orcid.org/0000-0002-2078-6835}{}
\author{Gonzalo Navarro}{Center for Biotechnology and Bioengineering (CeBiB) \\
Department of Computer Science, University of Chile,
Chile}{gnavarro@dcc.uchile.cl}{https://orcid.org/0000-0002-2286-741X}{Funded by Basal Funds FB0001, Mideplan,
Chile, and Fondecyt Grant 1-230755, Chile.}

\authorrunning{Zs.\ Lipt\'ak, F.\ Masillo, and G.\ Navarro} 
\Copyright{Zsuzsanna Lipt\'ak, Francesco Masillo, and Gonzalo Navarro}

\category{} 

\relatedversion{} 

\supplement{}

\nolinenumbers 
\hideLIPIcs

\EventLongTitle{}
\EventShortTitle{}
\EventAcronym{}
\EventYear{2024}
\EventDate{2024}
\EventLocation{}
\EventLogo{}

\begin{document}

\title{A Textbook Solution for Dynamic Strings}

\maketitle

\begin{abstract}
We consider the problem of maintaining a collection of strings while efficiently supporting splits and concatenations on them, as well as comparing two substrings, and computing the longest common prefix between two suffixes. 
This problem can be solved in optimal time $\Oh(\log N)$ whp for the updates and $\Oh(1)$ worst-case time for the queries, where $N$ is the total collection size [Gawrychowski et al., SODA 2018]. We present here a much simpler solution based on a forest of enhanced splay trees (\Fest), where both the updates and the substring comparison take $\Oh(\log n)$ amortized time, $n$ being the lengths of the strings involved.
The longest common prefix of length $\ell$ is computed in $\Oh(\log n + \log^2\ell)$ amortized time. Our query results are correct whp. Our simpler solution enables other more general updates in $\Oh(\log n)$ amortized time, such as reversing a substring and/or mapping its symbols. We can also regard substrings as circular or as their omega extension.

\end{abstract}



\section{Introduction}\label{sec:introduction}

Consider the problem in which we have to maintain a collection of {\em dynamic strings}, that is, strings we want to modify over time. The modifications may be edit operations such as insertion, deletion, or substitution of a single character; inserting or deleting an entire substring (possibly creating a new string from the deleted substring); adding a fresh string to the collection; etc. 
In terms of queries, we may want to retrieve a symbol or substring of a dynamic string, determine whether two substrings from anywhere in the collection are equal, or even determine the longest prefix shared by two suffixes in the collection (LCP). The collection must be maintained in such a way that both updates and queries have little cost.

This setup is known in general as the {\em dynamic strings} problem. A partial and fairly straightforward  solution are the so-called ropes, or cords~\cite{BoehmAP95}. These are binary trees\footnote{The authors~\cite{BoehmAP95} actually state that they are DAGs and referring to them as binary trees is just a simplification. The reason is that the nodes can have more than one parent, so subtrees may be shared.} where the leaves store short substrings, whose left-to-right concatenation forms the string. Ropes were introduced for the Cedar programming language to speed up handling very long strings; a C implementation (termed cords) was also given in the same paper~\cite{BoehmAP95}. 
As the motivating application of ropes/cords was that of implementing a text editor, they support edit operations and extraction/insertion of substrings to enable fast typing and cut\&paste, as well as
retrieving substrings, but do not support queries like substring equality or LCPs. The trees must be periodically rebalanced to maintain logarithmic times. Recently, a modified version of ropes was implemented for the Ruby language as a basic data type~\cite{MenardSD18}.
This variant supports the same updates but does not give any theoretical guarantee.

The first solution we know of that enables equality tests, by Sundar and Tarjan \cite{ST94}, supports splitting and concatenating whole sequences, and whole-string equality in constant time, with updates taking $\Oh(\sqrt{N\log m}+\log m)$ amortized time, where $N$ is the total length of all the strings in the collection and $m$ is the number of updates so far. It is easy to see that these three primitives encompass all the operations and queries above, except for LCP (substring retrieval is often implicit). The update complexity was soon improved by Mehlhorn et al.~\cite{MSU97} to $\Oh(\log^2 N)$ expected time with a randomized data structure, and $\Oh(\log N (\log m \log^*m + \log N))$ worst-case time with a deterministic one. The deterministic time complexity was later improved by Alstrup et al.~\cite{ABR00} to $\Oh(\log N \log^* N)$ (which holds with high probability, whp), also computing LCPs in $\Oh(\log N)$ worst-case time. Recently, Gawrychowski et al.~\cite{GawrychowskiKKL15,GawrychowskiKKL18} obtained $\Oh(\log N)$ update time whp, retaining constant time to compare substrings, and also decreasing the LCP time to constant, among many other results. They also showed that the problem is essentially closed because just updates and substring equality require $\Omega(\log N)$ time even if allowing amortization.  
Nishimoto et al.~\cite{NIIBT16,NishimotoIIBT20} showed how to compute LCPs in worst-case time $\Oh(\log N + \log\ell\log^* N)$, where $\ell$ is the LCP length, while inserting/deleting substrings of length $\ell$ in worst-case time $\Oh((\ell+\log N \log^* N)\frac{(\log\log N)^2}{\log\log\log N})$.

All these results build on the idea of parsing a string hierarchically by consistently cutting it into blocks, giving unique names to the blocks, and passing the sequence of names to the next level of parsing. The string is then represented by a parse tree of logarithmic height, whose root consists of a single name, which can be compared to the name at the root of another substring to determine string equality. While there is a general consensus on the fact that those solutions are overly complicated, 
Gawrychowski et al.~\cite{GawrychowskiKKL18} mention that

\vspace{1.5pt}
\textit{``We note that it is very simple to achieve $\Oh(\log n)$ update time [...], 
	if we allow the equality queries to give an incorrect
	result with polynomially small probability. We represent every string by a balanced search tree with
	characters in the leaves and every node storing a fingerprint of the sequence represented by its descendant
	leaves. However, it is not clear how to make the answers always correct in this approach [...]. 
	Furthermore, it seems that both computing the longest
	common prefix of two strings of length $n$ and comparing them lexicographically requires $\Omega(\log^2 n)$ time
	in this approach.''}
\vspace{1.5pt}

This suggestion, indeed, connects to the original idea of ropes~\cite{BoehmAP95}. Cardinal and Iacono \cite{CardinalI21} built on the suggestion to develop a kind of tree dubbed ``Data Dependent Tree (DDT)'', which enables updates and LCP computation in $\Oh(\log N)$ {\em expected amortized} time, yet with no errors. DDTs eliminate the chance of errors by ensuring that the fingerprints have no collisions---they simply rebuild all DDTs for all strings in the collection, using a new hash function, when this low-probability event occurs---and reduce the LCP complexity to $\Oh(\log N)$ by ensuring that subtrees representing the same string have the same shape (so one can descend in the subtrees of both strings synchronously).

In this paper we build on the same suggestion \cite{GawrychowskiKKL18}, but explore the use of another kind of tree---an enhanced splay tree---which yields a beautifully simple yet powerful data structure for maintaining dynamic string collections. We obtain logarithmic {\em amortized} update times for most operations (our cost to compute LCPs lies between logarithmic and squared-logarithmic, see later)
and our queries return correct answers whp. The ease of implementation of splay trees makes our solution attractive to be included in a textbook for undergraduate students.

An important consequence of using simpler data structures is that our space usage is $\Oh(N)$, whereas the solutions based on parsings require in addition $\Oh(\log N)$ space per update performed, as each one adds a new path to the parse tree. Since the previous parse tree is still available, those structures are {\em persistent}: one can access any previous version. Our solution is not persistent in principle, but we can make it persistent using $\Oh(\log n)$ extra space per update or query made so far (we cannot make direct use of the techniques of Driscoll et al.~\cite{DriscollSST86}). These add only $\Oh(1)$ amortized time to the operations. 

It would not be hard to obtain {\em worst-case} times instead of amortized ones, by choosing AVL, $\alpha$-balanced, or other trees that guarantee logarithmic height. One can indeed find the use of such binary trees for representing strings in the literature \cite{Ryt03,CLLPPSS05,Gaw11}. Our solution using splay trees has the key advantage of being very simple and easy to understand. The basic operations of splitting and concatenating strings, using worst-case balanced trees, imply attaching and detaching many subtrees, plus careful rebalancing, which is a nightmare to explain and implement.\footnote{As an example, an efficient implementation \cite{KL21} of Rytter's AVL grammar \cite{Ryt03} has over 10,000 lines of C++ code considering only their ``basic'' variant.} Knuth, for example, considered them too complicated to include in his book \cite[p.~473]{Knu98} {\em ``Deletion, concatenation, etc. It is possible to do many other things to balanced trees and maintain the balance, but the algorithms are sufficiently lengthly that the details are beyond the scope of this book.''} Instead, he says \cite[p.~478]{Knu98} {\em ``A much simpler self-adjusting data structure called a splay tree was developed subsequently [...] Splay trees, like the other kinds of balanced trees already mentioned, support the operations of concatenation and splitting as well as insertion and deletion, and in a particularly simple way.''}

\subparagraph*{Our contribution.}
We use a splay tree~\cite{SleatorT85}, enhanced with additional information, to represent each string in the collection, where all the nodes contain string symbols and Karp-Rabin-like fingerprints \cite{KR87,NP18} of the symbols in their subtree.
We refer to our data structure as a {\em forest of enhanced splay trees}, or \Fest. 
As we will see,
we can create new strings in $\Oh(n)$ time, extract substrings of length $\ell$ in $\Oh(\ell+\log n)$ time, perform updates and (correctly whp) compare substrings in $\Oh(\log n)$ time, where $n$ is the length of the strings involved---as opposed to the total length $N$ of all the strings---and the times are amortized (the linear terms are also worst-case). 
Further, we can compute LCPs correctly whp in amortized time $\Oh(\log n + \log^2 \ell)$, where 
$\ell$ is the length of the returned LCP. 

While our LCP time is $\Oh(\log^2 n)$ for long enough $\ell$, LCPs are usually much shorter than the suffixes. For example, in considerably general probabilistic models \cite{Szp93}, the maximum LCP value between {\em any} distinct suffixes of two strings of length $n$ is almost surely $\Oh(\log n)$, in which case our algorithm runs in $\Oh(\log n)$ amortized time.

The versatility of our \Fest\ data structure allows us to easily support other kinds of operations, such as reversing or complementing substrings, or both. We can thus implement the reverse complementation of a substring in a DNA or RNA sequence, whereby the substring is reversed and each character is replaced by its Watson-Crick complement. Substring reversal alone is used in classic problems on genome rearrangements where genomes are represented as sequences of genes, and have to be sorted by reversals (see, e.g., \cite{WattersonEHM82,BafnaP93,Caprara97,CapraraR02,OliveiraDD17,CerbaiF20}, to cite just 
a few). Note that chromosomes can be viewed either as permutations or as strings, when gene duplication is taken into account, see Fertin et al.~\cite{BookGenomeRearrangments}; our \Fest\ data structure accommodates both. We can also implement signed reversals~\cite{HannenhalliP95,Han06}, another model of evolutionary operation used in genome rearrangements. In general, we can combine reversals with any involution on the alphabet, of which signed or Watson-Crick complementation are only examples. In order to support these operations in $\Oh(\log n)$ amortized time, we only need to add new constant-space annotations, further enhancing our splay trees while retaining the running times for the other operations. The obvious solution of maintaining modified copies of the strings (e.g., reversed, complemented, etc.) is less attractive in practice due to the extra space and time needed to store and update all the copies.

\subparagraph*{Operations supported.}
We maintain a collection of strings of total length $N$ in $\Oh(N)$ space, and support the following operations, where we distinguish the basic string data type from dynamic strings (all times are amortized). We have not chosen a minimal set of primitives because reducing to primitives entails considerable performance overheads in practice, even if the asymptotic time complexities are not altered.

\begin{itemize}
	\item $\makestring(w)$ creates a dynamic string $s$ from a basic string $w$, in $\Oh(|s|)$ time. 
	\item $\access(s,i)$ returns the symbol $s[i]$ in $\Oh(\log |s|)$ time.
	\item $\retrieve(s,i,j)$ returns the basic string $w[1..j-i+1] = s[i..j]$, in $\Oh(|w|+\log |s|)$ time.
	\item $\repl(s,i,c)$, $\ins(s,i,c)$, and $\del(s,i)$ perform the basic edit operations on $s$: substituting $s[i]$ by character $c$, inserting $c$ at $s[i]$, and deleting $s[i]$, respectively, all in $\Oh(\log|s|)$ time. For appending $c$ at the end of $s$ one can use $\ins(s,|s|+1,c)$.
	\item $\introduce(s_1,i,s_2)$ inserts $s_2$ at position $i$ of $s_1$ (for $1\le i\le |s_1|+1$), converting $s_1$ to $s_1[..i-1] \cdot s_2 \cdot s_1[i..]$ and destroying $s_2$,
	in $\Oh(\log |s_1s_2|)$ time. 
	\item $\extract(s,i,j)$ creates dynamic string $s' = s[i..j]$, removing it from $s$, 
	in $\Oh(\log |s|)$ time.
	\item $\equal(s_1,i_1,s_2,i_2,\ell)$ determines the equality of substrings $s_1[i_1..i_1+\ell -1]$ and $s_2[i_2..i_2+\ell -1]$ in $\Oh(\log |s_1s_2|)$ time, correctly whp.
	\item $\lcp(s_1,i_1,s_2,i_2)$ computes the length $\ell$ of the longest common prefix between suffixes $s_1[i_1..]$ and $s_2[i_2..]$, in $\Oh(\log |s_1s_2| + \log^2 \ell)$ time, correctly whp, and also tells which suffix is lexicographically smaller. 
	\item $\reverse(s,i,j)$ reverses the substring $s[i..j]$ of $s$, in $\Oh(\log |s|)$ time.
	\item $\map(s,i,j)$ applies a fixed involution (a symbol mapping that is its own inverse) to all the symbols of $s[i..j]$, in $\Oh(\log |s|)$ time.
\end{itemize}

Our data structure also enables easy implementation of other features, such as handling circular strings. This is an important and emerging topic~\cite{ayad2017mars,CharalampopoulosKRPRWZ22,GrossiIJLSZ23,grossi2016circular,IliopoulosKRRWZ23}, as many current sequence collections, in particular in computational biology, consist of circular rather than linear strings. Recent data structures built for circular strings~\cite{BoucherCL0S21a,BoucherCL0S24}, based on the extended Burrows-Wheeler Transform (eBWT)~\cite{MantaciRRS07}, avoid the detour via the linearization and handle the circular input strings directly.
Finally, \Fest\ also allows queries on the omega extensions of strings, that is, on the infinite concatenation $s^{\omega}=s\cdot s \cdot s \cdots$. These occur, for example, in the context of the eBWT, which is based on the so-called omega-order.  
In Section~\ref{sec:circ-omega-summary} we will sketch how to handle circular strings and the omega extension of strings; a detailed description will be given in the full version of the paper.

\section{Basic concepts}\label{sec:basics}

\subparagraph*{Strings.}\label{sec:strings}

We use array-based notation for strings, indexing from 1, so a string $s$ is a finite sequence over a finite ordered alphabet $\Sigma$, written $s=s[1..n]=s[1]s[2]\cdots s[n]$, for some $n\geq 0$. We assume that the alphabet $\Sigma$ is integer. The length of $s$ is denoted $|s|$, and $\epsilon$ denotes the {\em empty string}, the unique string of length $0$. For $1\leq i,j\leq |s|$, we write $s[i..j]=s[i]s[i+1]\cdots s[j]$ for the {\em substring} from $i$ to $j$, where $s[i..j] = \epsilon$ if $i>j$. We write {\em prefixes} as $s[..i] = s[1..i]$ and {\em suffixes} as $s[i..]=s[i..|s|]$. Given two strings $s,t$, their concatenation is written $s\cdot t$ or simply $st$, and $s^k$ denotes the $k$-fold concatenation of $s$, with $s^0=\epsilon$. A substring (prefix, suffix) of $s$ is called {\em proper} if it does not equal $s$.

The {\em longest common prefix} (\LCP) of two strings $s$ and $t$ is defined as the longest string $u$ that is both a prefix of $s$ and $t$, and $\mathrm{lcp}(s,t) = |u|$ as its length. 
One can define the lexicographic order based on the lcp: $s <_{\lex} t$ if either $s$ is a proper prefix of $t$, or otherwise if $s[\ell+1]<t[\ell+1],$ where $\ell = \mathrm{lcp}(s,t)$.

\subparagraph*{Splay trees.}\label{sec:splaytrees}

The {\em splay tree} \cite{SleatorT85} is a binary search tree that guarantees that a sequence of insertions, deletions, and node accesses costs $\Oh(\log n)$ amortized time per operation on a tree of $n$ nodes that starts initially empty. In addition, splay trees support splitting and joining trees, both in $\Oh(\log n)$ amortized time, where $n$ is the total number of nodes involved in the operation.  

The basic operation of the splay tree is called $\splay(x)$, which moves a tree node $x$ to the root by a sequence of primitive rotations called zig, zig-zig, zig-zag, and their symmetric versions. 
Let $x(A,B)$ denote a tree rooted at $x$ with left and right subtrees $A$ and $B$, then the rotation zig-zig converts $z(y(x(A,B),C),D)$ into $x(A,y(B,z(C,D))$, while the rotation zig-zag converts $z(y(A,x(B,C)),D)$ into $x(y(A,B),z(C,D))$. Whether zig-zig or zig-zag (or their symmetric variant) is applied to $x$ depends on its relative position w.r.t.\ its grandparent. Note that both of these operations are composed by two edge rotations. Finally, operation zig, which is only applied if $x$ is a child of the root, converts $y(x(A,B),C)$ into $x(A,y(B,C))$. 

Every access or update on the tree is followed by a $\splay$ on the deepest reached node. 
In particular, after finding a node $x$ in a downward traversal, we do $\splay(x)$ to make $x$ the tree root. The goal is that the costs of all the operations are proportional to the cost of all the related  $\splay$ operations performed, so we can focus on analyzing only the splays.
Many of the splay tree properties can be derived from a general ``access lemma'' \cite[Lem.~1]{SleatorT85}.

\begin{lemma}[Access Lemma~\cite{SleatorT85}] \label{lem:access}
	Let us assign any positive weight $w(x)$ to the nodes $x$ of a splay tree $T$, and define $sw(x)$ as the sum of the weights of all the nodes in the subtree rooted at $x$. Then, the amortized time to splay $x$ is $\Oh(\log(W/sw(x))) \subseteq \Oh(\log(W/w(x)))$, where $W = \sum_{x \in T} w(x)$.
\end{lemma}

The result is obtained by defining $r(x) = \log sw(x)$ (all our logarithms are in base 2) and $\Phi(T) = \sum_{x \in T} r(x)$ as the potential function for the splay tree $T$. If we choose $w(x)=1$ for all $x$, then $W=n$ on a splay tree of $n$ nodes, and thus we obtain $\Oh(\log n)$ amortized cost for each operation. By choosing other functions $w(x)$, one can prove other properties of splay trees like static optimality, the static finger property, and the working set property \cite{SleatorT85}. 

The update operations supported by splay trees include inserting new nodes, deleting nodes, joining two trees (where all the nodes in the second tree go to the right of the nodes in the first tree), and splitting a tree into two at some node (where all the nodes to its right become a second tree). The times of those operations are ruled by the ``balance theorem with updates'' \cite[Thm.~6]{SleatorT85}.

\begin{lemma}[Balance Theorem with Updates~\cite{SleatorT85}] \label{lem:update}
	Any sequence of access, insert, delete, join and split operations on a collection of initially empty splay trees has an amortized cost of $\Oh(\log n)$ per operation, where $n$ is the size of the tree(s) where the operation is carried out.
\end{lemma}

This theorem is proved with the potential function that assigns $w(x)=1$ to every node $x$. Note the theorem considers a forest of splay trees, whose potential function is the sum of the functions $\Phi(T)$ over the trees $T$ in the forest. For details, see the original paper~\cite{SleatorT85}.

\subparagraph*{Karp-Rabin fingerprinting.} \label{sec:kr}

Our queries will be correct ``with high probability'' (whp), meaning a probability of at least $1-1/N^c$ for an arbitrarily large constant $c$, where $N$ is the total size of the collection. This will come from the use of a variant of the original Karp-Rabin fingerprint~\cite{KR87} (cf.~\cite{NP18})
defined as follows. Let $[1..a]$ be the alphabet of our strings and $p \ge a$ a prime number. We choose a random base $b$ uniformly from $[1..p-1]$. The
fingerprint $\kappa$ of string $s[1..n]$ is defined as $\kappa(s) = \left(\sum_{i=0}^{n-1}s[n-i] \cdot b^i\right) \bmod p$. 
We say that two strings $s\neq s'$ of the same length $n$ collide through $\kappa$ if $\kappa(s)=\kappa(s')$, that is, $\kappa(s'')=0$ 
where $s''=s-s'$ is the string defined by $s''[i]=(s[i]-s'[i]) \bmod p$.
Since $\kappa(s'')$ is a polynomial, in the variable $b$, of degree at most $n-1$ over 
the field $\mathbb Z_p$, it has at most $n-1$ roots. The probability of a collision between two strings of length $n$ is then bounded by $(n-1)/(p-1)$ because $b$ is uniformly chosen in $[1..p-1]$. By choosing $p \in \Theta(N^{c+1})$ for any desired constant $c$, we obtain that $\kappa$
is collision-free on any $s\neq s'$ whp.  
We will actually choose $p \in \Theta(N^{c+2})$ because some of our operations perform $\Oh(\mathrm{polylog}\, N)$ string comparisons, not just one. Since $N$ varies over time, we can use instead a fixed upper bound, like the total amount of main memory. We use the RAM machine model where  logical and arithmetic operations on $\Theta(\log N)$ machine words take constant time.

Two fingerprints $\kappa(s)$ and $\kappa(s')$ can then be composed in constant time to form $\kappa(s' \cdot s) = (\kappa(s') \cdot b^{|s|} + \kappa(s)) \bmod p$. To avoid the $\Oh(\log |s|)$ time for  modular exponentiation, we will  maintain the value $b^{|s|} \bmod p$ together with $\kappa(s)$. The corresponding value for $s' \cdot s$ is $(b^{|s'|} \cdot b^{|s|}) \bmod p$, so we can maintain those powers in constant time upon concatenations.

\section{Our data structure and standard operations}\label{sec:linear}

In this section we describe our data structure called \Fest\ (for Forest of enhanced Splay Trees), composed of a collection of (enhanced) splay trees, and then show how the traditional operations on dynamic strings are carried out on it.  

\subsection{The data structure} \label{sec:struct}

We will use a $\FST$ for maintaining the collection of strings, one splay tree per string. A dynamic string $s[1..n]$ is encoded in a splay tree with $n$ nodes such that $s[k]$ is stored in the node $x$ with in-order $k$ (the in-order of a node is the position in which it is listed if we recursively traverse first the left subtree, then the node, and finally the right subtree). We will say that node $x$ {\em represents} 
the substring $s[i..j]$, where $[i..j]$ is the range of the in-orders of all the nodes in the subtree rooted at $x$.
Let $T$ be the splay tree representing string $s$, then for $1\leq i \leq |s|$, we call $\node(i)$ the node with in-order $i$, and for a node $x$ of $T$, we call $\pos(x)$ the in-order of node $x$. 
The root of $T$ is denoted $\root(T)$.

For the amortized analysis of our \FST, our potential function $\Phi$ will be the sum of the potential functions $\Phi(T)$ over all the splay trees $T$ representing our string collection. The collection starts initially empty, with $\Phi=0$. New strings are added to the collection with $\makestring$; then edited with $\repl$, $\ins$, and $\del$, and redistributed with $\introduce$ and $\extract$.

\subparagraph*{Information stored at nodes.}
A node $x$ of the splay tree representing $s[i..j]$ will contain pointers to its left and right children, called $x\leftc$ and $x\rightc$, its symbol $x\car = s[\pos(x)]$, its subtree size $x\size = j-i+1$, its fingerprint $x\fp = \kappa(s[i..j])$, and the value $x\power = b^{j-i+1} \bmod p$. These fields are recomputed in constant time whenever a node $x$ acquires new children $x\leftc$ and/or $x\rightc$ (e.g., during the splay rotations) with the following formulas:  (1) $x\size = x\leftc\size + 1 + x\rightc\size$, (2) $x\fp = ((x\leftc\fp \cdot b + x\car) \cdot x\rightc\power + x\rightc\fp) \bmod p$, and (3) $x\power = (x\leftc\power \cdot b \cdot x\rightc\power) \bmod p$,
as explained in Section~\ref{sec:kr}. For the formula to be complete when the left and/or right child is \textit{null}, we assume $null\size=0$, $null\fp=0$, and $null\power=1$.
We will later incorporate other fields. 

Subtree sizes allow us identify $\node(i)$ given $i$, in the splay tree $T$ representing string $s$, in $\Oh(\log |s|)$ amortized time. This means we can answer $\access(s,i)$ in $\Oh(\log |s|)$ amortized time, since $s[i] = \node(i)\car$. Finding $\node(i)$ is done in the usual way, with the recursive function $\find(i)=$ $\find(\root(T),i)$ that returns the $i$th smallest element in the subtree rooted at the given node. More precisely, $\find(x,i) = x$ if $i = x\leftc\size+1$, $\find(x,i) = \find(x\leftc,i)$ if $i<x\leftc\size+1$, and $\find(x,i) = \find(x\rightc,i-(x\leftc\size+1))$ if $i > x\leftc\size+1$. To obtain logarithmic amortized time, $\find$ splays the node it returns, thus $\pos(\root(T))=i$ holds after calling $\find(\root(T),i)$. 

\subparagraph*{Isolating substrings.}
We will make use of another primitive we call $\isolate(i,j)$, for $1\leq i,j \leq |s|$ and $i \le j+1$, on a tree $T$ representing string $s$. This operation rearranges $T$ in such a way that $s[i..j]$ becomes represented by one subtree, and returns this subtree's root $y$. 

If $i=1$ and $j=n$, then $y=\root(T)$ and we are done. If $i=1$ and $j<n$, then we find (and splay) $\node(j+1)$ using $\find(j+1)$; 
this will move $\node(j+1)$ to the root, and $s[i..j]$ will be represented by the left subtree of the root, so $y=\root(T)\leftc$. Similarly, if $1<i$ and $j=n$, then we perform $\find(i-1)$, so $\node(i-1)$ is splayed to the root and $s[i..j]$ is represented by the right subtree of the root, thus $y=\root(T)\rightc$. 

Finally, if $1<i, j<n$, then splaying first $\node(j+1)$ and then $\node(i-1)$ will typically result in $\node(i-1)$ being the root and $\node(j+1)$ its right child, thus the left subtree of $\node(j+1)$ contains $s[i..j]$, that is, $y = \root(T)\rightc\leftc$. The only exception arises if the last splay operation on $\node(i-1)$ is a zig-zig, as in this case $\node(j+1)$ would become a grandchild, not a child, of the root. Therefore, in this case, we modify the last splay operation: if $\node(i-1)$ is a grandchild of the root and a zig-zig must be applied, we perform instead two consecutive zig operations on $\node(i-1)$ in a bottom-up manner, that is, we first rotate the edge between $\node(i-1)$ and its parent, and then the edge between $\node(i-1)$ and its new parent (former grandparent), see Fig.~\ref{fig:isolate}. 

We now consider the effect of the modified zig-zig operation on the potential. In the proof of Lemma~\ref{lem:access} \cite[Lem.~1]{SleatorT85}, Sleator and Tarjan show that the zig-zig and the zig-zag cases contribute $3(r'(x)-r(x))$ to the amortized cost, where $r'(x)$ is the new value of $r(x)$ after the operation. The sum then telescopes to $3(r(t)-r(x)) = 3\log(sw(t)/sw(x))$ along an upward path towards a root node $t$. The zig rotation, instead, contributes $1+r'(x)-r(x)$, where the $1$ would be problematic if it was not applied only once in the path. Our new zig-zig may, at most one time in the path, cost like two zig's, $2+2(r'(x)-r(x))$, which raises the cost bound
of the whole splay operation from $1+3\log(sw(t)/sw(x))$ to $2+3\log(sw(t)/sw(x))$. This retains the amortized complexity, that is, the amortized time for $\isolate$ is $\Oh(\log |s|)$.

\input{modified_zigzig}

\subsection{Creating a new dynamic string} \label{sec:makestring}

Given a basic string $w[1..n]$, operation $\makestring(w)$ creates a new dynamic string $s[1..n]$ with the same content as $w$, which is added to the \FST. While this can be accomplished in $\Oh(n\log n)$ amortized time via successive $\ins$ operations on an initially empty string, we describe a ``bulk-loading'' technique that achieves linear worst-case (and amortized) time.

The idea is to create, in $\Oh(n)$ time, a perfectly balanced splay tree using the standard recursive procedure. As we show in the next lemma, this shape of the tree adds only $\Oh(n)$ to the potential function, and therefore the amortized time of this procedure is also $\Oh(n)$. 

\begin{lemma}\label{lemma:balanced}
	The potential $\Phi(T)$ of a perfectly balanced splay tree $T$ with $n$ nodes is at most $2n+\Oh(\log^2 n) \subseteq \Oh(n)$.
\end{lemma}
\begin{proof}
	Let $d$ be the depth of the deepest leaves in a perfectly balanced binary tree, and call $l = d-d'+1$ the {\em level} of any node of depth $d'$. It is easy to see that there are at most $1+n/2^l$ subtrees of level $l$. Those subtrees have at most $2^l-1$ nodes. 
	Separating the sum $\Phi(T) = \sum_{x\in T} r(x)$ by levels $l$ and using the bound $sw(x) < 2^l$ if $x$ is of level $l$, we get
	$\Phi(T) < \sum_{l=1}^{\log n} \left(1+\frac{n}{2^l}\right) \log 2^l = 
	2n + \Oh(\log^2 n).$
\end{proof}

Once the tree is created and the fields $x\car$ are assigned in in-order, we perform a post-order traversal to compute the other fields. This is done in constant time per node using the formulas given in Section~\ref{sec:struct}.

\subsection{Retrieving a substring} \label{sec:retrieve}

Given a string $s$ in the \FST\ and two indices $1 \leq i \leq j \leq |s|$, operation $\retrieve(s,i,j)$ extracts the substring $s[i..j]$ and returns it as a basic string. 
The special case $i=j$ is given by $\access(s,i)$, which finds $\node(i)$, splays it, and returns $\root(T)\car$, recall Section~\ref{sec:struct}. 
If $i<j$, we perform $y=\isolate(i,j)$ and then we return $s[i..j]$ with an in-order traversal of the subtree rooted at $y$. 
Overall, the operation $\retrieve(s,i,j)$ takes $\Oh(\log |s|)$ amortized time for $\isolate$, and then $\Oh(j-i+1)$ worst case time for the traversal of the subtree.

\subsection{Edit operations}

Let $s$ be a string in the \FST, $i$ an index of $s$, and $c$ a character. The simplest edit operation, $\repl(s,i,c)$ writes $c$ at $s[i]$, that is, $s$ becomes $s'=s[..i-1]\cdot c\cdot s[i+1..]$. It is implemented by doing $\find(i)$ in the splay tree $T$ of $s$, in $\Oh(\log |s|)$ amortized time. After the operation, $\node(i)$ is the root, so we set $\root(T)\car = c$ and recompute (only) its fingerprint as explained in Section~\ref{sec:struct}.

Now consider operation $\ins(s,i,c)$, which converts $s$ into $s'=s[..i-1]\cdot c\cdot s[i..]$. This corresponds to the standard insertion of a node in the splay tree, at in-order position $i$.
We first use $\find(i)$ in order to make $x = \node(i)$ the tree root, and then create a new root node $y$, with $y\leftc = x\leftc$ and $y\rightc = x$. We then set  $x\leftc = null$ and recompute the other fields of $x$ as shown in Section~\ref{sec:struct}. Finally, we set $y\car=c$ and also compute its other fields.
By Lemma~\ref{lem:update}, the amortized cost for an insertion is $\Oh(\log |s|)$.

Finally, the operation $\del(s,i)$ converts $s$ into $s' = s[..i-1] \cdot s[i+1..]$. This corresponds to standard deletion in the splay tree: we first do $\find(i)$ in the tree $T$ of $s$, so that $x = \node(i)$ becomes the root, and then join the splay trees of $x\leftc$ and $x\rightc$, isolating the root node $x$ and freeing it. The joined tree now represents $s'$; the amortized cost is $\Oh(\log |s|)$.

\subsection{Introducing and extracting substrings}

Given two strings $s_1$ and $s_2$ represented by trees $T_1$ and $T_2$ in the \FST, and an insertion position $i$ in $s_1$, operation $\introduce(s_1,i,s_2)$ generates a new string $s = s_1[..i-1] \cdot s_2 \cdot s_1[i..]$ (the original strings are not anymore available). We implement this operation by first doing $y = \isolate(i,i-1)$ on the tree $T_1$. Note that in this case $y$ will be a $null$ node, whose in-order position is between $i-1$ and $i$. We then replace this null node by (the root of) the tree $T_2$.  
As shown in Section~\ref{sec:struct}, the node $y$ that we replace has at most two ancestors in $T_1$, say $x_1$ (the root) and $x_2$. We must then recompute the fields of $x_2$ and then of $x_1$. 

Apart from the $\Oh(\log |s_1|)$ amortized time for $\isolate$, the other operations take constant time. We must consider the change in the potential introduced by connecting $T_2$ to $T_1$. In the potential $\Phi$, the summands $\log sw(x_1)$ and $\log sw(x_2)$ will increase to 
$\log (sw(x_1)+|s_2|)$ and $\log (sw(x_2)+|s_2|)$, thus the increase is $\Oh(\log |s_2|)$.
The total amortized time is thus $\Oh(\log |s_1|+\log|s_2|) = \Oh(\log |s_1s_2|)$.

Let $s$ be a string represented by tree $T$ in the \FST\ and $i \le j$ indices in $s$. Function $\extract(s,i,j)$ removes $s[i..j]$ from $s$ and creates a new dynamic string $s'$ from it. This can be carried out by first doing $y = \isolate(i,j)$ on $T$, then detaching $y$ from its parent in $T$ to make it the root of the tree that will represent $s'$, and finally recomputing the fields of the (former) ancestors $x_2$ and $x_1$ of $y$. The change in potential is negative, as $\log sw(x_1)$ and $\log sw(x_2)$ decrease by up to $\Oh(\log(j-i+1))$. The total amortized time is then $\Oh(\log |s|)$.

\subsection{Substring equality} \label{sec:equal}

Let $s_1[i_1..i_1+\ell -1]$ and $s_2[i_2..i_2+\ell -1]$ be two substrings, where possibly $s_1=s_2$. Per Section~\ref{sec:kr}, we can compute $\equal$ whp by comparing  $\kappa(s_1[i_1..i_1+\ell-1])$ and $\kappa(s_2[i_2..i_2+\ell-1])$. 
We compute $y_1 = \isolate(i_1,i_1+\ell-1)$ on the tree of $s_1$ and $y_2 = \isolate(i_2,i_2+\ell-1)$ on the tree of $s_2$. Once node $y_1$ represents $s_1[i_1..i_1+\ell-1]$ and $y_2$ represents $s_2[i_2..i_2+\ell-1]$, we compare $y_1\fp = \kappa(s_1[i_1..i_1+\ell-1])$ with $y_2\fp = \kappa(s_2[i_2..i_2+\ell-1])$.

The splay operations take $\Oh(\log|s_1s_2|)$ amortized time, while the comparison of the fingerprints takes constant time and returns the correct answer whp. Note this is a one-sided error; if the method answers negatively, the strings are distinct.

\section{Extended operations}

In this section we consider less standard operations of dynamic strings, including the computation of LCPs and others we have not seen addressed before.

\subsection{Longest common prefixes} \label{sec:lcp}

Operation $\lcp(s_1,i_1,s_2,i_2)$ computes $\mathrm{lcp}(s_1[i_1..],s_2[i_2..])$ correctly whp, by exponentially searching for the maximum value $\ell$ such that $s_1[i_1..i_1+\ell-1] = s_2[i_2..i_2+\ell-1]$. The exponential search requires $\Oh(\log \ell)$ equality tests, which are done using $\equal$ operations. The amortized cost of this basic solution is then $\Oh(\log |s_1s_2| \log \ell)$; we now improve it. 

We note that all the accesses the exponential search performs in $s_1$ and $s_2$ are at distance $\Oh(\ell)$ from $s_1[i_1]$ and $s_2[i_2]$. We could then use the dynamic finger property \cite{Cole00} to show, with some care, that the amortized time is $\Oh(\log |s_1s_2| + \log^2\ell)$. This property, however, uses a different mechanism of potential functions where trees cannot be joined or split.\footnote{The static finger property cannot be used either, because we need new fingers every time an LCP is computed. Extending the ``unified theorem'' \cite[Thm.~5]{SleatorT85} to $m$ fingers (to support $m$ LCP operations in the sequence) introduces an $\Oh(\log m)$ additive amortized time in the operations, since now $W=\Theta(m)$.} We then use an alternative approach. The main idea is that, if we could bound $\ell$ beforehand, we could isolate those areas so that the accesses inside them would cost $\Oh(\log \ell)$ and then we could reach the desired amortized time. Bounding $\ell$ in less than $\Oh(\log\ell)$ accesses (i.e., $\Oh(\log |s_1s_2| \log \ell)$ time) is challenging, however.
Assuming for now that $s_1 \not= s_2$ (we later handle the case $s_1=s_2$), our plan is as follows (see Fig.~\ref{fig:lcp}):

\begin{enumerate}
	\item Find a (crude) upper bound $\ell' \ge \ell$.
	\item Extract substrings $s_1' = s_1[i_1..i_1+\ell'-1]$ and $s_2' = s_2[i_2..s_2+\ell'-1]$.
	\item Run the basic exponential search for $\ell$ between $s_1'[1..]$ and $s_2'[1..]$.
	\item Reinsert substrings $s_1'$ and $s_2'$ into $s_1$ and $s_2$.
\end{enumerate}

Steps 2 and 4 are carried out in $\Oh(\log |s_1s_2|)$ amortized time using the operations $\extract$ and $\introduce$, respectively. Step 3 will still require $\Oh(\log\ell)$ substring comparisons, but since they will be carried out on the shorter substrings $s_1'$ and $s_2'$, they will take $\Oh(\log \ell \log \ell')$ amortized time. The main challenge is to balance the cost to find $\ell'$ in Step 1 with the quality of the approximation of $\ell'$ so that $\log\ell'$ is not much larger than $\log\ell$.

\begin{figure}[t]

\resizebox{\textwidth}{!}{

\tikzset{every picture/.style={line width=0.75pt}} 

\begin{tikzpicture}[x=0.75pt,y=0.75pt,yscale=-1,xscale=1]

\draw   (30,64) -- (58,142) -- (2,142) -- cycle ;
\draw   (86,11) .. controls (86,4.92) and (90.92,0) .. (97,0) .. controls (103.08,0) and (108,4.92) .. (108,11) .. controls (108,17.08) and (103.08,22) .. (97,22) .. controls (90.92,22) and (86,17.08) .. (86,11) -- cycle ;
\draw    (94,22) -- (30,64) ;
\draw   (108,43) .. controls (108,36.92) and (112.92,32) .. (119,32) .. controls (125.08,32) and (130,36.92) .. (130,43) .. controls (130,49.08) and (125.08,54) .. (119,54) .. controls (112.92,54) and (108,49.08) .. (108,43) -- cycle ;
\draw   (98,74) -- (126,130) -- (70,130) -- cycle ;
\draw    (100,22) -- (118,32) ;
\draw    (98,74) -- (116,54) ;
\draw   (156,74) -- (182,138) -- (130,138) -- cycle ;
\draw    (122,54) -- (156,74) ;
\draw   (70.17,148.42) .. controls (70.14,153.09) and (72.46,155.43) .. (77.13,155.45) -- (87.97,155.51) .. controls (94.64,155.55) and (97.96,157.9) .. (97.94,162.57) .. controls (97.96,157.9) and (101.3,155.59) .. (107.97,155.62)(104.97,155.61) -- (118.82,155.68) .. controls (123.49,155.71) and (125.83,153.39) .. (125.86,148.72) ;
\draw    (188,66) -- (286,66) ;
\draw [shift={(288,66)}, rotate = 180] [color={rgb, 255:red, 0; green, 0; blue, 0 }  ][line width=0.75]    (10.93,-3.29) .. controls (6.95,-1.4) and (3.31,-0.3) .. (0,0) .. controls (3.31,0.3) and (6.95,1.4) .. (10.93,3.29)   ;
\draw   (339,34) -- (367,90) -- (311,90) -- cycle ;
\draw   (311.17,94.42) .. controls (311.14,99.09) and (313.46,101.43) .. (318.13,101.45) -- (328.97,101.51) .. controls (335.64,101.55) and (338.96,103.9) .. (338.94,108.57) .. controls (338.96,103.9) and (342.3,101.59) .. (348.97,101.62)(345.97,101.61) -- (359.82,101.68) .. controls (364.49,101.71) and (366.83,99.39) .. (366.86,94.72) ;
\draw    (386,68) -- (484,68) ;
\draw [shift={(486,68)}, rotate = 180] [color={rgb, 255:red, 0; green, 0; blue, 0 }  ][line width=0.75]    (10.93,-3.29) .. controls (6.95,-1.4) and (3.31,-0.3) .. (0,0) .. controls (3.31,0.3) and (6.95,1.4) .. (10.93,3.29)   ;
\draw   (528,40) -- (556,114) -- (500,114) -- cycle ;
\draw   (500.31,144) .. controls (500.3,148.67) and (502.62,151.01) .. (507.29,151.02) -- (555.14,151.12) .. controls (561.81,151.14) and (565.13,153.48) .. (565.12,158.15) .. controls (565.13,153.48) and (568.47,151.16) .. (575.14,151.17)(572.14,151.16) -- (622.98,151.28) .. controls (627.65,151.29) and (629.99,148.96) .. (630,144.29) ;
\draw   (602,40) -- (630,96) -- (574,96) -- cycle ;
\draw   (554,11) .. controls (554,4.92) and (558.92,0) .. (565,0) .. controls (571.08,0) and (576,4.92) .. (576,11) .. controls (576,17.08) and (571.08,22) .. (565,22) .. controls (558.92,22) and (554,17.08) .. (554,11) -- cycle ;
\draw    (562,22) -- (528,40) ;
\draw    (568,22) -- (602,40) ;
\draw   (500.31,118.42) .. controls (500.29,123.09) and (502.61,125.43) .. (507.28,125.45) -- (525.13,125.52) .. controls (531.8,125.55) and (535.12,127.89) .. (535.1,132.56) .. controls (535.12,127.89) and (538.46,125.57) .. (545.13,125.6)(542.13,125.59) -- (562.97,125.68) .. controls (567.64,125.7) and (569.98,123.38) .. (570,118.71) ;
\draw    (568,180) -- (568,198) -- (96,198) -- (96,181) ;
\draw [shift={(96,179)}, rotate = 90] [color={rgb, 255:red, 0; green, 0; blue, 0 }  ][line width=0.75]    (10.93,-3.29) .. controls (6.95,-1.4) and (3.31,-0.3) .. (0,0) .. controls (3.31,0.3) and (6.95,1.4) .. (10.93,3.29)   ;

\draw (92,162) node [anchor=north west][inner sep=0.75pt]   [align=left] {$\displaystyle \ell '$};
\draw (187,70) node [anchor=north west][inner sep=0.75pt]  [font=\scriptsize] [align=left] {$\mathtt{extract}( s_{1} ,i_{1} ,i_{1} {+}\ell'{-}1)$};
\draw (333,108) node [anchor=north west][inner sep=0.75pt]   [align=left] {$\displaystyle \ell '$};
\draw (387,71) node [anchor=north west][inner sep=0.75pt]  [font=\scriptsize] [align=left] {exponential search};
\draw (561.12,157.58) node [anchor=north west][inner sep=0.75pt]   [align=left] {$\displaystyle \ell '$};
\draw (529.52,132) node [anchor=north west][inner sep=0.75pt]   [align=left] {$\displaystyle \ell $};
\draw (66,132) node [anchor=north west][inner sep=0.75pt]  [font=\scriptsize] [align=left] {$\displaystyle i_{1}$};
\draw (305,187) node [anchor=north west][inner sep=0.75pt]  [font=\scriptsize] [align=left] {re-introduce};

\end{tikzpicture}
}

\caption{Scheme of operations for \lcp\ shown on one of the two strings.}
\label{fig:lcp}

\end{figure}
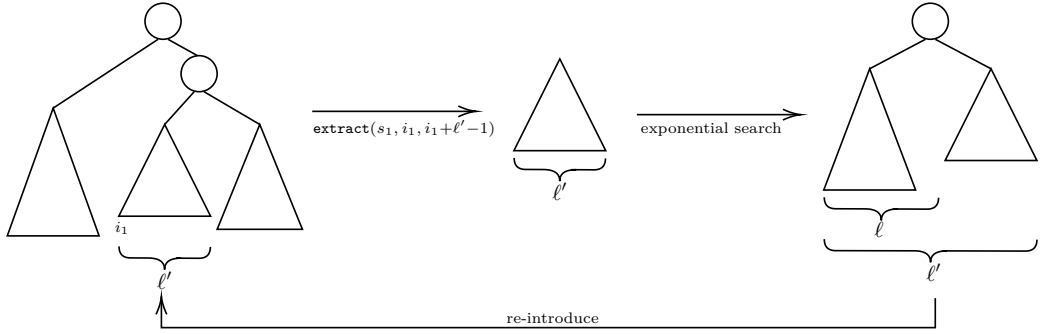

Consider the following strategy for Step 1. Let $n=|s_1s_2|$ and $n'=\min(|s_1|-i_1+1,|s_2|-i_2+1)$. We first check a few border cases that we handle in $\Oh(\log n)$ amortized time: if $s_1[i_1..i_1+n'-1] = s_2[i_2..i_2+n'-1]$ we finish with the answer $\ell = n'$, or else if $s_1[i_1..i_1+1] \not= s_2[i_2..i_2+1]$ we finish with the answer $\ell=0$ or $\ell=1$. Otherwise, we define the sequence $\ell_0 = 2$ and $\ell_i = \min(n',\ell_{i-1}^{\,2})$ and try out the values $\ell_i$ for $i = 1,2,\ldots$, until we obtain $s_1[i_1..i_1+\ell_i-1] \not= s_2[i_2..i_2+\ell_i-1]$. This implies that $\ell_{i-1} \le \ell < \ell_i$, so we can use $\ell' = \ell_i \le \ell^2$. This yields $\Oh(\log\ell \log\ell') = \Oh(\log^2 \ell)$ amortized time for Step 3. On the other hand, since $\ell \ge \ell_{i-1} = 2^{2^{i-1}}$, it holds $i \le 1+\log\log \ell$. Since each of the $i$ values is tried out in $\Oh(\log n)$ time with $\equal$, the amortized cost of Step 1 is $\Oh(\log n \log\log \ell)$ and the total cost to compute $\lcp$ is $\Oh(\log n \log\log \ell + \log^2 \ell)$. In particular, this is $\Oh(\log^2\ell)$ when $\ell$ is large enough, $\log \ell = \Omega(\sqrt{\log n \log\log n})$.

\subparagraph*{Hitting twice.}
To obtain our desired time $\Oh(\log n + \log^2\ell)$ for every value of $\log\ell$, we will apply our general strategy twice. First, we will set $\ell'' = 2^{\log^{2/3} n}$ and determine whether $s_1[i_1..i_1+\ell''-1] = s_2[i_2..i_2+\ell''-1]$. If they are equal, then $\log \ell = \Omega(\log^{2/3} n)$ and we can apply the strategy of the previous paragraph verbatim, obtaining amortized time $\Oh(\log^2 \ell)$. If they are not equal, then we know that $\ell'' > \ell$, so we $\extract$ $s_1'' = s_1[i_1..i_1+\ell''-1]$ and $s_2'' = s_2[i_2..i_2+\ell''-1]$ to complete the search for $\ell'$ inside those (note we are still in Step 1). 
We use the same sequence $\ell_i$ of the previous paragraph, with the only difference that the accesses are done on trees of size $\ell''$ and not $n$; therefore each step costs $\Oh(\log\ell'')=\Oh(\log^{2/3} n)$ instead of $\Oh(\log n)$. After finally finding $\ell'$, we $\introduce$ back $s_1''$ and $s_2''$ into $s_1$ and $s_2$. Step 1 then completes in amortized time $\Oh(\log n + \log^{2/3} n \log\log \ell) = \Oh(\log n)$. Having found $\ell' \le \ell^2$, we proceed with Step 2 onwards as above, taking $\Oh(\log^2\ell)$ additional time.

\subparagraph*{When the strings are the same.}
In the case $s_1=s_2$, assume w.l.o.g.\ $i_1 < i_2$. We can still carry out Step 1 and, if $i_1+\ell' \le i_2$, proceed with the plan in the same way, extracting $s_1'$ and $s_2'$ from the same string and later reintroducing them. In case $i_1+\ell' > i_2$, however, both substrings overlap. In this case we extract just one substring, $s' = s_1[i_1..i_2+\ell'-1]$, which is of length at most $2\ell'$, and run the basic exponential search between $s'[1..]$ and $s'[i_2-i_1+1..]$ still in amortized time $\Oh(\log\ell \log\ell')$. We finally reintroduce $s'$ in $s_1$. The same is done if we need to extract $s_1''$ and $s_2''$: if both come from the same string and $i_1+\ell'' > i_2$, then we extract just one single string $s''=s[i_1..i_2+\ell''-1]$ and obtain the same asymptotic times.

\subparagraph*{Lexicographic comparisons.}
Once we know that (whp) the LCP of the suffixes is of length $\ell$, we can determine which is smaller by accessing (using $\access$) the symbols at positions $s_1[i_1+\ell]$ %
and $s_2[i_2+\ell]$ and comparing them, in $\Oh(\log |s_1s_2|)$ additional amortized time.

\subsection{Substring reversals} \label{sec:reversals}

Operation $\reverse(s,i,j)$ changes $s$ to $s[..i-1] s[j] s[j-1] \cdots s[i+1]s[i] s[j+1..]$.
Reflecting it directly in our current structure requires $\Omega(j-i+1)$ time, which is potentially $\Omega(|s|)$. Our strategy, instead, is to just ``mark'' the subtrees where the reversal should be carried out, and de-amortize its cost across future operations, materializing it progressively as we traverse the marked subtrees. To this end, 
we extend our \FST\ data structure with a new Boolean field $x\rev$ in each node $x$, which indicates that its whole subtree should be regarded as reversed, that is, its descending nodes should be read right-to-left, but that this update has not yet been carried out. This field is set to \textit{false} on newly created nodes. We also add a field $x\fprev$, so that if $x$ represents $s[i..j]$, then $x\fprev = \kappa(s[j]s[j-1]\cdots s[i+1]s[i])$ is the fingerprint of the reversed string. When $x\rev$ is \textit{true}, the fields of $x$ (including $x\fp$ and $x\fprev$) still do not reflect the reversal.

The fields $x\fprev$ must be maintained in the same way as the fields $x\fp$. Concretely, upon every update where the children of node $x$ change, we not only recompute $x\fp$ as shown in Section~\ref{sec:struct}, but also $x\fprev = ((x\rightc\fprev \cdot b + x\car) \cdot x\leftc\power + x\leftc\fprev) \bmod p$.

In order to apply the described reversal to a substring $s[i..j]$, we first compute $y = \isolate(i,j)$ on the tree of $s$, and then toggle the Boolean value $y\rev = \neg~ y\rev$ (note that, if $y$ had already an unprocessed reversal, this is undone without ever materializing it). The operation $\reverse$ then takes $\Oh(\log |s|)$ amortized time, dominated by the cost of $\isolate(i,j)$. We must, however, handle potentially reversed nodes. 

\subparagraph*{Fixing marked nodes.}
Every time we access a tree node, if it is marked as reversed, we {\em fix} it, after which it can be treated as a regular node because its fields will already reflect the reversal of its represented string (though some descendant nodes may still need fixing). 

Fixing a node involves exchanging its left and right children, toggling their reverse marks, and updating the node fingerprint. 
More precisely, we define the primitive $\fix(x)$ as follows: if $x\rev$ is $\mathit{true}$, then (i) set $x\rev = \mathit{false}$, $x\leftc\rev = \neg~x\leftc\rev$, $x\rightc\rev = \neg~x\rightc\rev$, (ii) swap $x\leftc$ with $x\rightc$, and (iii) swap $x\fp$ with $x\fprev$. See Fig.~\ref{fig:fix-operations} for an example. 
It is easy to see that $\fix$ maintains the invariants about the meaning of the reverse fields. 

Because all the operations in splay trees, including the $\splay$, are done along paths that are first traversed downwards from the root, it suffices that we run $\fix(x)$ on every node $x$ we find as we descend from the root (for example, on every node $x$ where we perform $\find(x,i)$), before taking any other action on the node. This ensures that all the accesses and structural changes to the splay tree are performed over fixed nodes, and therefore no algorithm needs further changes. For example, when we perform $\splay(x)$, all the ancestors of $x$ are already fixed. As another example, if we run $\equal$ as in Section~\ref{sec:equal}, the nodes $y_1$ and $y_2$ will already be fixed by the time we read their fingerprint fields. As a third example, if we run $\retrieve(s,i,j)$ as in Section~\ref{sec:retrieve} and the subtree of $y$ has reversed nodes inside, we will progressively fix all those nodes as we traverse the subtree, therefore correctly retrieving $s[i..j]$ within $\Oh(j-i+1)$ time.

\begin{figure}[t]
    \centering

\resizebox{\textwidth}{!}{

\tikzset{every picture/.style={line width=0.75pt}} 

\begin{tikzpicture}[x=0.75pt,y=0.75pt,yscale=-1,xscale=1]

\draw   (43.92,40.25) .. controls (43.92,33.01) and (49.58,27.14) .. (56.55,27.14) .. controls (63.52,27.14) and (69.18,33.01) .. (69.18,40.25) .. controls (69.18,47.5) and (63.52,53.37) .. (56.55,53.37) .. controls (49.58,53.37) and (43.92,47.5) .. (43.92,40.25) -- cycle ;
\draw   (21.96,77.32) -- (43.92,156) -- (0,156) -- cycle ;
\draw    (54.9,53.37) -- (21.96,77.32) ;
\draw   (87.84,77.32) -- (109.8,156) -- (65.88,156) -- cycle ;
\draw    (56.55,53.37) -- (87.84,77.32) ;
\draw   (334.61,40.25) .. controls (334.61,33.01) and (340.26,27.14) .. (347.24,27.14) .. controls (354.21,27.14) and (359.86,33.01) .. (359.86,40.25) .. controls (359.86,47.5) and (354.21,53.37) .. (347.24,53.37) .. controls (340.26,53.37) and (334.61,47.5) .. (334.61,40.25) -- cycle ;
\draw   (312.65,77.32) -- (334.61,156) -- (290.69,156) -- cycle ;
\draw    (345.59,53.37) -- (312.65,77.32) ;
\draw   (378.53,77.32) -- (400.49,156) -- (356.57,156) -- cycle ;
\draw    (347.24,53.37) -- (378.53,77.32) ;
\draw   (617.31,40.25) .. controls (617.31,33.01) and (622.97,27.14) .. (629.94,27.14) .. controls (636.92,27.14) and (642.57,33.01) .. (642.57,40.25) .. controls (642.57,47.5) and (636.92,53.37) .. (629.94,53.37) .. controls (622.97,53.37) and (617.31,47.5) .. (617.31,40.25) -- cycle ;
\draw   (595.35,77.32) -- (617.31,156) -- (573.39,156) -- cycle ;
\draw    (628.29,53.37) -- (595.35,77.32) ;
\draw   (661.24,77.32) -- (683.2,156) -- (639.27,156) -- cycle ;
\draw    (629.94,53.37) -- (661.24,77.32) ;
\draw    (444,104) -- (552,104) ;
\draw [shift={(554,104)}, rotate = 180] [color={rgb, 255:red, 0; green, 0; blue, 0 }  ][line width=0.75]    (10.93,-3.29) .. controls (6.95,-1.4) and (3.31,-0.3) .. (0,0) .. controls (3.31,0.3) and (6.95,1.4) .. (10.93,3.29)   ;
\draw    (135,104) -- (246,104) ;
\draw [shift={(248,104)}, rotate = 180] [color={rgb, 255:red, 0; green, 0; blue, 0 }  ][line width=0.75]    (10.93,-3.29) .. controls (6.95,-1.4) and (3.31,-0.3) .. (0,0) .. controls (3.31,0.3) and (6.95,1.4) .. (10.93,3.29)   ;

\draw (51.05,33.54) node [anchor=north west][inner sep=0.75pt]   [align=left] {$x$};
\draw (16.11,124.19) node [anchor=north west][inner sep=0.75pt]   [align=left] {$A$};
\draw (81.84,124.19) node [anchor=north west][inner sep=0.75pt]   [align=left] {$B$};
\draw (73,27.14) node [anchor=north west][inner sep=0.75pt] [align=left] {{\scriptsize $x\rev = \textit{true}$}};
\draw (135,123) node [anchor=north west][inner sep=0.75pt]  [font=\footnotesize] [align=left] {toggle rev};
\draw (341.74,34.54) node [anchor=north west][inner sep=0.75pt]   [align=left] {$x$};
\draw (306.79,124.19) node [anchor=north west][inner sep=0.75pt]   [align=left] {$A$};
\draw (372.53,124.19) node [anchor=north west][inner sep=0.75pt]   [align=left] {$B$};
\draw (235,52) node [anchor=north west][inner sep=0.75pt]  [font=\scriptsize] [align=left] {{\footnotesize $\root(A)\rev =$}\\{\footnotesize  $\neg \root(A)\rev$}};
\draw (445,127) node [anchor=north west][inner sep=0.75pt]  [font=\footnotesize] [align=left] {swap left-right\\and $x\fp$ with $x\fprev$};
\draw (624.44,34.54) node [anchor=north west][inner sep=0.75pt]   [align=left] {$x$};
\draw (589.35,124.19) node [anchor=north west][inner sep=0.75pt] [align=left] {$B$};
\draw (655.38,124.19) node [anchor=north west][inner sep=0.75pt] [align=left] {$A$};
\draw (363,29) node [anchor=north west][inner sep=0.75pt]  [align=left] {{\scriptsize $x\rev = \textit{false}$}};
\draw (388,54) node [anchor=north west][inner sep=0.75pt]  [font=\scriptsize] [align=left] {{\footnotesize $\root(B)\rev =$}\\{\footnotesize  $\neg \root(B)\rev$}};
\draw (544,28) node [anchor=north west][inner sep=0.75pt]  [align=left] {{\scriptsize $x\fp\leftrightarrow x\fprev$}};

\end{tikzpicture}
}
    \caption{Scheme of the $\fix$ operation on node $x$.}
    \label{fig:fix-operations}
\end{figure}
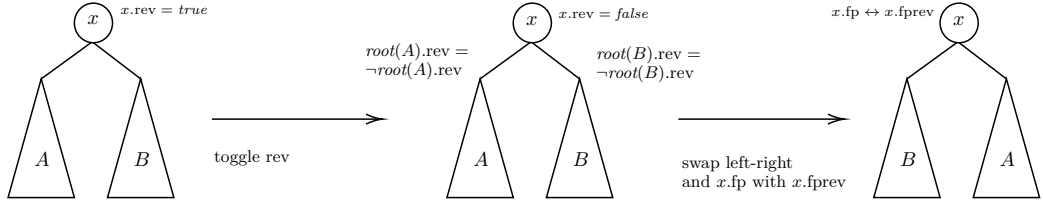

Note that $\fix$ takes constant time per node and does not change the potential function $\Phi$, so no time complexities change due to our adjustments. The new fields also enable other queries, for example to decide whether a string is a palindrome.

\subsection{Involutions}

We support the operation $\map(s,i,j)$ analogously to substring reversals, that is, isolating $s[i..j]$ in a node $y = \isolate(i,j)$ and then marking that the substring covered by node $y$ is mapped using a new Boolean field $y\mapf$, which is set to \textit{true}. This will indicate that every symbol $s[k]$, for $i \le k \le j$, must be interpreted as $f(s[k])$, but that the change has not yet been materialized. Similarly to \reverse, this information will be propagated downwards as we descend into a subtree, otherwise it is maintained in the subtree's root only.
The operation will then take $\Oh(\log |s|)$ amortized time.

To manage the mapping and deamortize its linear cost across subsequent operations, we will also store fields $x\mfp = \kappa(f(s[i]) f(s[i+1]) \cdots f(s[j]))$ and $x\mfprev = \kappa(f(s[j]) f(s[j-1]) \cdots f(s[i]))$, which maintain the fingerprint of the mapped string, and its reverse, represented by $x$. Those are maintained analogously as the previous fingerprints: (1) $x\mfp = ((x\leftc\mfp \cdot b + f(x\car)) \cdot x\rightc\power + x\rightc\mfp) \bmod p$, and (2) $x\mfprev = ((x\rightc\mfprev \cdot b + f(x\car)) \cdot x\leftc\power + x\leftc\mfprev) \bmod p$.

As for string reversals, every time we access a tree node, if it is marked as mapped, we unmark it and toggle the mapped mark of its children, before proceeding with any other action. Precisely, we define the primitive $\fixm(x)$ as follows: if $x\mapf$ is $\mathit{true}$, then (i) set $x\mapf = \mathit{false}$, $x\leftc\mapf = \neg~x\leftc\mapf$, $x\rightc\mapf = \neg~x\rightc\mapf$, (ii) set $x\car = f(x\car)$, and (iii) swap $x\fp$ with $x\mfp$, and $x\fprev$ with $x\mfprev$. We note that, in addition, the $\fix$ operation defined in Section~\ref{sec:reversals} must also exchange $x\mfp$ with $x\mfprev$ if we also support involutions.
Note how, as for reversals, two applications of $f$ cancel each other, which is correct because $f$ is an involution. Operation $\fixm$ is applied in the same way as $\fix$ along tree traversals. 

\subparagraph*{Reverse complementation.}
By combining string reversals and involutions, we can for example support the application of {\em reverse complementation} of substrings in DNA sequences, where a substring $s[i..j]$ is reversed and in addition its symbols are replaced by their Watson-Crick complement, applying the involution $f(\mathtt{A}) = \mathtt{T}$, $f(\mathtt{T}) = \mathtt{A}$, $f(\mathtt{C}) = \mathtt{G}$, and $f(\mathtt{G}) = \mathtt{C}$. In case we {\em only} want to perform reverse complementation (and not reversals and involutions independently), we can simplify our fields and maintain only a Boolean field $x\rc$ and the fingerprint $x\mfprev$ in addition to $x\fp$. Fixing a node  consists of: if $x\rc$ is $\mathit{true}$, then (i) set $x\rc = \mathit{false}$, $x\leftc\rc = \neg~x\leftc\rc$, $x\rightc\rc = \neg~x\rightc\rc$, (ii) set $x\car = f(x\car)$, (iii) swap $x\leftc$ with $x\rightc$,  (iv) swap $x\fp$ with $x\mfprev$. 

\section{Circular strings and omega extension}\label{sec:circ-omega-summary}

Our data structure can be easily extended to handle circular strings. We do this by introducing a new routine, called \ttrotate, which allows us linearize the circular string starting at any of its indices. By carefully using this primitive, along with a slight modification for the computation of fingerprints, we can support every operation that we presented on linear strings with the same time bounds, as well as signed reversals, in $\Oh(\log|\hat{s}|)$ amortized time. 

By supporting operations on circular strings, we can also handle the omega extension of strings, which is the infinite concatenation of a string: $s^{\omega} = s \cdot s \cdots$. Again, we are able to meet the same time bounds on every operation on linear strings. We also define two ways to implement the equality between omega-extended substrings. Full details will be contained in the full version of the paper. 

\section{Conclusion}\label{sec:conclusion}

We presented a new data structure, a forest of enhanced splay trees (\FST), to handle collections of dynamic strings. Our solution is much simpler than those offering the best theoretical results, 
while still offering logarithmic amortized times for most update and query operations. We answer queries correctly whp, and updates are always correct.

To build our data structure, we employ an approach that differs from theoretical solutions: we use a splay tree for representing each string, enhancing it with additional annotations. The use of binary trees to represent dynamic strings is not new, but exploiting the simplicity of splay trees for attaching and detaching subtrees is. As our \FST\ is easy to understand, explain, and implement, we believe that it offers the opportunity of wide usability and can become a textbook implementation of dynamic strings. Further, we have found nontrivial---yet perfectly implementable---solutions to relevant queries, like computing the length $\ell$ of the longest common prefix of two suffixes in time $\Oh(\log n + \log^2 \ell)$ instead of the trivial $\Oh(\log^2 n)$. The simplicity of our solution enables new features, like the possibility of reversing a substring, or reverse-complementing it, to be easily implemented in logarithmic amortized time. 
Our data structure also allows handling circular strings, as well as omega-extensions of strings---features competing solutions have not explored.

\bibliography{biblio}

\newpage
\appendix

\section*{APPENDIX}

\section{Figures}

\begin{figure}[H]
    \centering

    \resizebox{\textwidth}{!}{
    \tikzset{every picture/.style={line width=0.75pt}} 
    
    \begin{tikzpicture}[x=0.75pt,y=0.75pt,yscale=-1,xscale=1]
    
    \draw    (89,22) -- (47,30) ;
    \draw    (141,30) -- (99,22) ;
    \draw    (41,52) -- (23,68) ;
    \draw    (137,52) -- (119,68) ;
    \draw    (19,90) -- (11,106) ;
    \draw    (115,90) -- (107,106) ;
    \draw    (51,52) -- (71,68) ;
    \draw    (147,52) -- (165,68) ;
    \draw    (67,90) -- (59,106) ;
    \draw    (163,90) -- (155,106) ;
    \draw    (183,68) -- (290,68) ;
    \draw [shift={(292,68)}, rotate = 180] [color={rgb, 255:red, 0; green, 0; blue, 0 }  ][line width=0.75]    (10.93,-3.29) .. controls (6.95,-1.4) and (3.31,-0.3) .. (0,0) .. controls (3.31,0.3) and (6.95,1.4) .. (10.93,3.29)   ;
    \draw    (318,94) -- (308,109) ;
    \draw    (303,131) -- (295,147) ;
    \draw    (328,94) -- (338,109) ;
    \draw    (408,22) -- (426,38) ;
    \draw    (334,131) -- (326,147) ;
    \draw    (424,60) -- (416,76) ;
    \draw    (336,57) -- (326,72) ;
    \draw    (348,58) -- (360,72) ;
    \draw    (364,21) -- (346,37) ;
    \draw    (443,68) -- (554,68) ;
    \draw [shift={(556,68)}, rotate = 180] [color={rgb, 255:red, 0; green, 0; blue, 0 }  ][line width=0.75]    (10.93,-3.29) .. controls (6.95,-1.4) and (3.31,-0.3) .. (0,0) .. controls (3.31,0.3) and (6.95,1.4) .. (10.93,3.29)   ;
    \draw    (620,20) -- (652,28) ;
    \draw    (643,126) -- (653,141) ;
    \draw    (649,163) -- (641,179) ;
    \draw    (632,88) -- (641,104) ;
    \draw    (648,49) -- (630,66) ;
    \draw    (579,28) -- (608,20) ;
    \draw    (599,88) -- (591,104) ;
    \draw    (680,88) -- (672,104) ;
    \draw    (660,50) -- (680,66) ;
    \draw    (580,50) -- (600,66) ;
    
    \draw (39,200) node [anchor=north west][inner sep=0.75pt]   [align=left] {$s =$mississippi};
    \draw    (11.5, 117.5) circle [x radius= 11.34, y radius= 11.34]   ;
    \draw (6,112) node [anchor=north west][inner sep=0.75pt]  [font=\scriptsize] [align=left] {m};
    \draw    (94.5, 11.5) circle [x radius= 11.34, y radius= 11.34]   ;
    \draw (89,6) node [anchor=north west][inner sep=0.75pt]  [font=\scriptsize] [align=left] {$\:$s};
    \draw    (46.5, 41.5) circle [x radius= 11.34, y radius= 11.34]   ;
    \draw (41,36) node [anchor=north west][inner sep=0.75pt]  [font=\scriptsize] [align=left] {$\:$s};
    \draw    (142.5, 41.5) circle [x radius= 11.34, y radius= 11.34]   ;
    \draw (137,36) node [anchor=north west][inner sep=0.75pt]  [font=\scriptsize] [align=left] {$\:$p};
    \draw    (23.5, 79.5) circle [x radius= 11.34, y radius= 11.34]   ;
    \draw (18,74) node [anchor=north west][inner sep=0.75pt]  [font=\scriptsize] [align=left] {$\:$i};
    \draw    (70.5, 79.5) circle [x radius= 11.34, y radius= 11.34]   ;
    \draw (65,74) node [anchor=north west][inner sep=0.75pt]  [font=\scriptsize] [align=left] {$\:$i};
    \draw    (58.5, 117.5) circle [x radius= 11.34, y radius= 11.34]   ;
    \draw (53,112) node [anchor=north west][inner sep=0.75pt]  [font=\scriptsize] [align=left] {$\:$s};
    \draw    (119.5, 79.5) circle [x radius= 11.34, y radius= 11.34]   ;
    \draw (114,74) node [anchor=north west][inner sep=0.75pt]  [font=\scriptsize] [align=left] {$\:$i};
    \draw    (107.5, 117.5) circle [x radius= 11.34, y radius= 11.34]   ;
    \draw (102,112) node [anchor=north west][inner sep=0.75pt]  [font=\scriptsize] [align=left] {$\:$s};
    \draw    (166.5, 79.5) circle [x radius= 11.34, y radius= 11.34]   ;
    \draw (161,74) node [anchor=north west][inner sep=0.75pt]  [font=\scriptsize] [align=left] {$\:$i};
    \draw    (154.5, 117.5) circle [x radius= 11.34, y radius= 11.34]   ;
    \draw (149,112) node [anchor=north west][inner sep=0.75pt]  [font=\scriptsize] [align=left] {$\:$p};
    \draw (180,72) node [anchor=north west][inner sep=0.75pt]   [align=left] {$\mathtt{extract}( s,9,11)$};
    \draw (341,200) node [anchor=north west][inner sep=0.75pt]   [align=left] {$s =$mississi\\$s'=$ppi};
    \draw    (295.5, 158.5) circle [x radius= 11.34, y radius= 11.34]   ;
    \draw (290,153) node [anchor=north west][inner sep=0.75pt]  [font=\scriptsize] [align=left] {m};
    \draw    (342.5, 47.5) circle [x radius= 11.34, y radius= 11.34]   ;
    \draw (337,42) node [anchor=north west][inner sep=0.75pt]  [font=\scriptsize] [align=left] {$\:$s};
    \draw    (323.5, 83.5) circle [x radius= 11.34, y radius= 11.34]   ;
    \draw (318,78) node [anchor=north west][inner sep=0.75pt]  [font=\scriptsize] [align=left] {$\:$s};
    \draw    (403.5, 11.5) circle [x radius= 11.34, y radius= 11.34]   ;
    \draw (398,6) node [anchor=north west][inner sep=0.75pt]  [font=\scriptsize] [align=left] {$\:$p};
    \draw    (307.5, 120.5) circle [x radius= 11.34, y radius= 11.34]   ;
    \draw (302,115) node [anchor=north west][inner sep=0.75pt]  [font=\scriptsize] [align=left] {$\:$i};
    \draw    (337.5, 120.5) circle [x radius= 11.34, y radius= 11.34]   ;
    \draw (332,115) node [anchor=north west][inner sep=0.75pt]  [font=\scriptsize] [align=left] {$\:$i};
    \draw    (325.5, 158.5) circle [x radius= 11.34, y radius= 11.34]   ;
    \draw (320,153) node [anchor=north west][inner sep=0.75pt]  [font=\scriptsize] [align=left] {$\:$s};
    \draw    (370.5, 11.5) circle [x radius= 11.34, y radius= 11.34]   ;
    \draw (365,6) node [anchor=north west][inner sep=0.75pt]  [font=\scriptsize] [align=left] {$\:$i};
    \draw    (362.5, 83.5) circle [x radius= 11.34, y radius= 11.34]   ;
    \draw (357,78) node [anchor=north west][inner sep=0.75pt]  [font=\scriptsize] [align=left] {$\:$s};
    \draw    (427.5, 49.5) circle [x radius= 11.34, y radius= 11.34]   ;
    \draw (422,44) node [anchor=north west][inner sep=0.75pt]  [font=\scriptsize] [align=left] {$\:$i};
    \draw    (415.5, 87.5) circle [x radius= 11.34, y radius= 11.34]   ;
    \draw (410,82) node [anchor=north west][inner sep=0.75pt]  [font=\scriptsize] [align=left] {$\:$p};
    \draw (440,72) node [anchor=north west][inner sep=0.75pt]   [align=left] {$\mathtt{introduce}( s,1,s')$};
    \draw    (613.5, 10.5) circle [x radius= 11.34, y radius= 11.34]   ;
    \draw (608,5) node [anchor=north west][inner sep=0.75pt]  [font=\scriptsize] [align=left] {m};
    \draw    (626.5, 77.5) circle [x radius= 11.34, y radius= 11.34]   ;
    \draw (621,72) node [anchor=north west][inner sep=0.75pt]  [font=\scriptsize] [align=left] {$\:$i};
    \draw    (652.5, 152.5) circle [x radius= 11.34, y radius= 11.34]   ;
    \draw (647,147) node [anchor=north west][inner sep=0.75pt]  [font=\scriptsize] [align=left] {$\:$i};
    \draw    (640.5, 190.5) circle [x radius= 11.34, y radius= 11.34]   ;
    \draw (635,185) node [anchor=north west][inner sep=0.75pt]  [font=\scriptsize] [align=left] {$\:$s};
    \draw    (654.5, 39.5) circle [x radius= 11.34, y radius= 11.34]   ;
    \draw (649,34) node [anchor=north west][inner sep=0.75pt]  [font=\scriptsize] [align=left] {$\:$s};
    \draw    (639.5, 115.5) circle [x radius= 11.34, y radius= 11.34]   ;
    \draw (634,110) node [anchor=north west][inner sep=0.75pt]  [font=\scriptsize] [align=left] {$\:$s};
    \draw (567,220) node [anchor=north west][inner sep=0.75pt]   [align=left] {$s = \text{ppimississi}$};
    \draw    (576.5, 39.5) circle [x radius= 11.34, y radius= 11.34]   ;
    \draw (571,34) node [anchor=north west][inner sep=0.75pt]  [font=\scriptsize] [align=left] {$\:$p};
    \draw    (602.5, 77.5) circle [x radius= 11.34, y radius= 11.34]   ;
    \draw (597,72) node [anchor=north west][inner sep=0.75pt]  [font=\scriptsize] [align=left] {$\:$i};
    \draw    (590.5, 115.5) circle [x radius= 11.34, y radius= 11.34]   ;
    \draw (585,110) node [anchor=north west][inner sep=0.75pt]  [font=\scriptsize] [align=left] {$\:$p};
    \draw    (683.5, 77.5) circle [x radius= 11.34, y radius= 11.34]   ;
    \draw (678,72) node [anchor=north west][inner sep=0.75pt]  [font=\scriptsize] [align=left] {$\:$i};
    \draw    (671.5, 115.5) circle [x radius= 11.34, y radius= 11.34]   ;
    \draw (666,110) node [anchor=north west][inner sep=0.75pt]  [font=\scriptsize] [align=left] {$\:$s};

    \end{tikzpicture}
    }
    \caption{Cycle-rotation operation: {\ttrotate}$(s,9)$ moves $s[9..]$ to the left of $s[..8]$. After the rotation the string becomes $s[9..]s[..8]$.}
	\label{fig:rotation}
\end{figure}
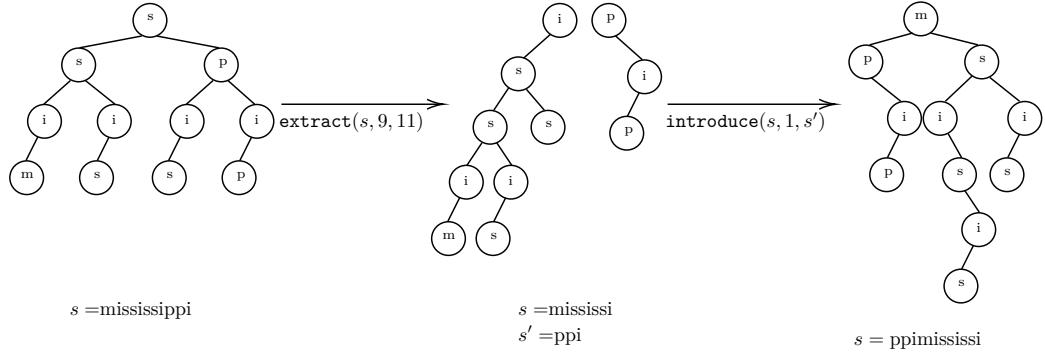
\tikzset{every picture/.style={line width=0.75pt}} 

\section{Other Related Work}
\label{sec:more-related}

    A related line of work aims at maintaining a data structure such that the solution to some particular problem on one or two strings can be efficiently updated when these strings undergo an edit operation (deletion, insertion, or substitution). Examples are longest common factor of two strings~\cite{AmirCIPR17,AmirCPR20}, optimal alignment of two strings~\cite{CharalampopoulosKM20}, approximating the edit distance~\cite{Kociumaka0S23}, longest palindromic substring~\cite{FunakoshiNIBT18}, longest square~\cite{AmirBCK19}, or longest Lyndon factor~\cite{UrabeNIBT18} of one string. The setup can be what is referred to as partially dynamic, when the original string or strings are returned to their state before the edit, or fully dynamic, when the edit operations are reflected on the original string or strings. Clifford et al.~\cite{CliffordGLS18} give lower bounds on various problems of this kind when a single substitution is applied. 
	
	This setup, also referred to as {\em dynamic strings}, differs from ours in several
	ways: (a) we are not only interested in solving one specific problem on strings; (b) we have an entire collection of strings, and will want to ask queries on any one or any pair of these; and (c) we allow many different kinds of update operations. 

 Locally consistent parsings to maintain dynamic strings have been used to support more complex problems, such as simulating suffix arrays~\cite{KempaK22,KempaK22arxiv}.


\section{Circular strings and omega extensions}\label{sec:circular-omega}

\subsection{Additional definitions}
In this section, we are going to use some further concepts regarding periodicity and conjugacy.

A string $s$ is called {\em periodic} with period $r$ if $s[i+r] = s[i]$ for all $1\leq i \leq |s|-r$.

Two strings $s,t$ are {\em conjugates} if there exist strings $u,v$, possibly empty, such that $s=uv$ and $t=vu$. Conjugacy is an equivalence; the equivalence classes $[s]$ are also called {\em circular strings}, and any $t\in [s]$ is called a {\em linearization} of this circular string. Abusing notation, any linear string $s$ can be viewed as a circular string, in which case it is taken as a representative of its conjugacy class. A {\em substring} of a circular string $s$ is any prefix of any $t\in [s]$, or, equivalently, a string of the form $s[i..j]$ for $1\leq i, j\leq |s|$ (a linear substring), or $s[i..]s[..j]$, where $j<i$. A {\em necklace} is a string $s$ with the property that $s\leq_{\lex} t$ for all $t\in [s]$. Every conjugacy class contains exactly one necklace. 

When the dynamic strings in our collection are to be interpreted as circular strings, we need to adjust some of our operations. Our model is that we will maintain a canonical representative $\hat{s}$ of the class of rotations of $s$. All the indices of the operations refer to positions in $\hat{s}$. Internally, we may store in the \FST\ another representative $s$ of the class, not necessarily $\hat{s}$.

\subsection{Circular strings}\label{sec:circular}

Our general approach to handle operations on $\hat{s}$ regarding it as circular is to rotate it conveniently before accessing it. The splay tree $T$ of $\hat{s}$ will then maintain some (string) rotation $s = \hat{s}[r..] \hat{s}[..r-1]$ of $\hat{s}$, and we will maintain a field $\start(\hat{s}) = r$ so that we can map any index $\hat{s}[i]$ referred to in update or query operations to $s[((|s|+i-\start(\hat{s})) \bmod |s|)+1]$.

When we want to change the rotation of $\hat{s}$ to another index $r'$, so that we now store $s' = \hat{s}[r'..] \hat{s}[..r'-1]$, we make use of a new operation $\ttrotate(s,i)$, which rotates $s$ so that its splay tree represents $s[i..]s[..i-1]$. This is implemented as $s' = \extract(s,i,|s|)$ followed by $\introduce(s,1,s')$. 
We then move from rotation $r$ to $r'$ in $\Oh(\log |s|)$ amortized time by doing $\ttrotate(s,r'-r+1)$ if $r'>r$, or $\ttrotate(s,|s|+r'-r+1)$ if $r'<r$. We then set $\start(\hat{s})=r'$.

Operation $s=\makestring(w)$ stays as before, in the understanding that $\hat{s}=w$ will be seen as the canonical representation of the class, so we set $\start(\hat{s})=1$; this can be changed later with a string rotation if desired. All the operations that address a single position $\hat{s}[i]$, like $\access$ and the edit operations, are implemented verbatim by just shifting the index $i$ using $\start(\hat{s})$ as explained. Instead, the operations $\retrieve$, $\extract$, $\equal$, $\reverse$, and $\map$, which act on a range $\hat{s}[i..j]$, may give trouble when $i>j$, as in this case the substring is $\hat{s}[i..] \hat{s}[..j]$ by circularity. In this case, those operations will be preceded by a change of rotation from the current one, $r=\start(\hat{s})$, to $r'=1$, using $\ttrotate$ as explained. This guard will get rid of those cases. Note that, in the case of $\equal$, we may need to rotate both $s_1$ and $s_2$, independently, to compute each of the two signatures.

The two remaining operations deserve some consideration. Operation $\introduce(s_1,i,s_2)$ could be implemented verbatim (with the shifting of $i$), but in this case it would introduce in $\hat{s_1}[i]$ the current rotation of $s_2$, instead of $\hat{s_2}$ as one would expect. Therefore, we precede the operation by a change of rotation in $s_2$ to $r'=1$, which makes the splay tree store $\hat{s_2}$ with $\start(\hat{s_2})=1$. 

Finally, in operation $\lcp(s_1,i_1,s_2,i_2)$ we do not know for how long the LCP will extend, so we precede it by changes of rotations in both $s_1$ and $s_2$ that make them start at position $1$ of $\hat{s_1}$ and $\hat{s_2}$. In case $s_1=s_2$, however, this trick cannot be used. One simple solution is to rotate the string every time we call $\equal$ during Step 1; recall Section~\ref{sec:lcp}. This will be needed as long as the accesses are done on $s_1$ and $s_2$; as soon as we extract the substrings of length $\ell''$ (and, later, $\ell'$ for Step 3), we work only on the extracted strings. While the complexity is preserved, rotating the string every time can be too cumbersome. We can use an alternative way to compute signatures of circular substrings, $\kappa(s[i..]s[..j])$: we compute as in Section~\ref{sec:equal} $\sigma = \kappa(s[i..])$ and $\tau = \kappa(s[..j])$, as well as $b^j\bmod p$, which comes for free with the computation of $\tau$; then $\kappa(s[i..]s[..j]) = (\sigma\cdot b^j + \tau) \bmod p$.

Overall, we maintain for all the operations the same asymptotic running times given in the Introduction when the strings are interpreted as circular.

\subparagraph*{Signed reversals on circular strings.} 
By combining reversals and involutions, we can support signed reversals on circular strings, too. We do this in the same way as for linear strings, namely by doubling the alphabet $\Sigma$ of gene identifiers such that each gene $i$ has a negated version $-i$, and using the involution $f(i)=-i$ (and $f(-i)=i$). Note that the original paper in which reversals were introduced~\cite{WattersonEHM82} used circular chromosomes.

\subsection{Omega extensions}\label{sec:omega}

Circular dynamic strings allow us to implement operations that act on the omega extensions of the underlying strings. Recall that for a (linear) string $s$, the infinite string $s^{\omega}$ is defined as the infinite concatenation $s^{\omega} = s\cdot s \cdot s \cdots$. These are, for example, used in the definition of the {\em extended Burrows-Wheeler Transform} (eBWT) of Mantaci et al.~\cite{MantaciRRS07}, where the underlying string order is based on omega extensions. In this case, comparisons of substrings may need to be made whose length exceeds the shorter of the two strings $s_1$ and $s_2$. We therefore introduce a generalization of circular substrings as follows: $t$ is called an {\em omega-substring} of $s$ if $t = s[i..]s^ks[..j]$ for some $j<i-1$  and $k\geq 0$. Note that the suffix $s[i..]$ and the prefix $s[..j]$ may also be empty. Thus, $t$ is an omega-substring of $s$ if and only if $t = v^kv[..j]$ for some $k\geq 1$ and some conjugate $v$ of $s$.

An important tool in this section will be the 
famous Fine and Wilf Lemma~\cite{lothaire3}, which states that if a string $w$ has two periods $r,q$ and $|w|\geq r+q-\gcd(r,q)$, then $w$ is also periodic with period $\gcd(r,q)$ (a string $s$ is called periodic with period $r$ if $s[i+r] = s[i]$ for all $1\leq i \leq |s|-r$). 
The following is a known corollary, a different formulation of which was proven, e.g., in~\cite{MantaciRRS07}; we reprove it here for completeness. 

\begin{lemma} \label{lem:periodic}
    Let $u,v$ be two strings. If $\mathrm{lcp}(u^{\omega}, v^{\omega}) \geq |u|+|v|-\gcd(|u|,|v|)$, then $u^{\omega} = v^{\omega}$. 
\end{lemma}

\begin{proof}
    Let $\ell = \mathrm{lcp}(u^{\omega}, v^{\omega}) \geq |u|+|v|-\gcd(|u|,|v|)$. Then the string $t = s_1^{\omega}[..\ell]$ is periodic both with period $|u|$ and with period $|v|$, and thus, by the Fine and Wilf lemma, it is also periodic with period $\gcd(|u|,|v|)$. Since $\gcd(|u|,|v|) \leq |u|,|v|$, this implies that both $u$ and $v$ are powers of the same string $x$, of length $\gcd(|u|,|v|)$ and therefore, $u^{\omega} = x^{\omega} = v^{\omega}$. 
\end{proof}

We further observe that the fingerprint of strings of the form $u^k$ 
can be computed from the fingerprint of string $u$. More precisely, let $u$ be a string, $\pi = \kappa(u)$ its fingerprint, and $k\geq 1$. Then, calling $d = b^{|u|} \bmod p$ (which we also obtain in the field $y\power$ when computing $\kappa(u)$), it holds 
\begin{eqnarray}
\kappa(u^k) &=& (\pi\cdot d^{k-1} + \pi\cdot d^{k-2} + \cdots + \pi\cdot d + \pi ) \bmod p \nonumber\\ 
 &=& (\pi \cdot (d^{k-1} + d^{k-2} + \cdots + 1)) \bmod p, \label{eq:geo}
\end{eqnarray}
where $\geomsum(d,k-1)=(d^{k-1} + d^{k-2} + \cdots + 1) \bmod p$ can be computed in $\Oh(\log k)$ time using the identity $d^{2k+1}+d^{2k}+\cdots+1 = (d+1)\cdot((d^2)^k + (d^2)^{k-1} + \cdots+1)$, as follows\footnote{This technique seems to be folklore. Note that the better known formula $\geomsum(d,k)=((d^{k+1}-1)\cdot(d-1)^{-1}) \bmod p$ requires computing multiplicative inverses, which takes $\Oh(\log N)$ time using the extended Euclid's algorithm, or $\Oh(\log\log N)$ with faster algorithms \cite{SZ04}; those terms would not be absorbed by others in our cost formula.} (all modulo $p$):
\begin{eqnarray}
\geomsum(d,0) &=& 1 \nonumber \\
\geomsum(d,2k+1) &=& (d+1) \cdot \geomsum(d^2,k)  \label{eq:matrix} \\
\geomsum(d,2k) &=& d\cdot \geomsum(d,2k-1) + 1 \nonumber 
\end{eqnarray}

\subparagraph*{Extended substring equality.} We devise at least two ways in which our $\equal$ query can be extended to omega extensions.
First, consider the query $\equal_\omega(s_1,i_1,s_2,i_2,\ell) = \equal(s_1^\omega,i_1,s_2^\omega,i_2,\ell)$, that is, the normal substring equality interpreted on the omega extensions of $s_1$ and $s_2$.
We let $v_1 = \ttrotate(s_1,i_1)$ and $v_2 = \ttrotate(s_2,i_2)$. Then we have $s_1^{\omega}[i_1..i_1+\ell-1] = v_1^{k_1}v_1[..j_1]$, 
where $k_1 = \lfloor\ell/|s_1|\rfloor$ and $j_1 = \ell \bmod |s_1|$. If $k_1=0$, we simply compute $\kappa_1 = \kappa(s_1^{\omega}[i_1..i_1+\ell-1]) = \kappa(v_1[..j_1])$. Otherwise, we compute $\kappa_1 = \kappa(s_1^{\omega}[i_1..i_1+\ell-1])$ by applying Eq.~\eqref{eq:geo} as follows: 
\begin{eqnarray} \label{eq:powers}
\kappa_1 = (\kappa(v_1) \cdot (d^{k_1-1}+\cdots+1)\cdot b^{j_1} + \kappa(v_1[..j_1])) \bmod p.  
\end{eqnarray}

There are various components to compute in this formula apart from the fingerprints themselves. First, note that
$d = b^{|s_1|} \bmod p = b^{|v_1|} \bmod p = \root(T_1)\power$ for the tree $T_1$ of $s_1$ (or $v_1$), so we have it in constant time. Second, $b^{j_1} \bmod p$ is the field $y\power$ after we compute $\kappa(v_1[..j_1])$ via $y=\isolate(v_1,1,j_1)$ after completion of \ttrotate($s_1,i_1$), thus we also have it in constant time.
Third, 
$d^{k_1-1}+\cdots+1 = \geomsum(d,k_1-1)$ is computed with Eq.~(\ref{eq:matrix}) in time $\Oh(\log k_1) \subseteq \Oh(\log\ell)$.

By Lemma~\ref{lem:periodic} we can define $\ell_\omega=|s_1|+|s_2|$ and, if $\ell \geq \ell_\omega$, run the $\equal_\omega$ 
query with $\ell_\omega$ instead of $\ell$. The lemma shows that $s_1[i_1..i_1+\ell-1] = s_2[i_2..i_2+\ell-1]$ iff $s_1[i_1..i_1+\ell_\omega-1] = s_2[i_2..i_2+\ell_\omega-1]$. This limits $\ell$ to $|s_1|+|s_2|$ in our query and therefore the cost $\Oh(\log\ell)$ is in $\Oh(\log |s_1s_2|)$.

We compute $\kappa_2$ analogously, and return {\em true} if and only if $\kappa_1 = \kappa_2$, after undoing the rotations to get back the original strings $s_1$ and $s_2$. The total amortized time for operation $\equal_\omega$ is then $\Oh(\log |s_1s_2|)$.
Note that our results still hold whp because we are deciding on fingerprints of strings of length $\Oh(N)$, not $\Oh(\ell)$ (which is in principle unbounded).

A second extension of $\equal$ is $\equal_\omega^\omega(s_1,i_1,\ell_1,s_2,i_2,\ell_2)$, interpreted as $(s_1^\omega[i_1..i_1+\ell_1-1])^\omega = (s_2^\omega[i_2..i_2+\ell_2-1])^\omega$, that is, the omega extension of $s_1^\omega[i_1..i_1+\ell_1-1]$ is equal to the omega extension of $s_2^\omega[i_2..i_2+\ell_2-1]$. By Lemma~\ref{lem:periodic}, this is equivalent to $(s_1^\omega[i_1..i_1+\ell_1-1])^{\ell_2} = (s_2^\omega[i_2..i_2+\ell_2-1])^{\ell_1}$. So we first compute $\kappa_1=\kappa(s_1^\omega[i_1..i_1+\ell_1-1])$ and $\kappa_2=\kappa(s_2^\omega[i_2..i_2+\ell_2-1])$ as above, compute $d_1 = b^{\ell_1} \bmod p$ and $d_2 = b^{\ell_2} \bmod p$, and then return whether $(\kappa_1 \cdot (d_1^{\ell_2-1}+\cdots+1)) \bmod p =
(\kappa_2 \cdot (d_2^{\ell_1-1}+\cdots+1)) \bmod p$.
Operation $\equal_\omega^\omega$ is then also computed in amortized time $\Oh(\log |s_1s_2|)$.

\subparagraph*{Extended longest common prefix.}
We are also able to extend LCPs to omega extensions: operation $\lcp_\omega(s_1,i_1,s_2,i_2)$ computes, for the corresponding rotations $v_1=\ttrotate(s_1,i_1)$ and $v_2=\ttrotate(s_2,i_2)$, the longest common prefix length lcp($v_1^{\omega},v_2^{\omega}$), as well as the lexicographic order of $v_1^{\omega}$ and $v_2^{\omega}$. That this can be done efficiently follows again from Lemma~\ref{lem:periodic}. We first compare their omega-substrings of length $\ell_\omega = |s_1|+|s_2|$. If $\equal_\omega(s_1,i_1,s_2,i_2,\ell_\omega)$ answers \textit{true}, then it follows that $\lcp(s_1,i_1,s_2,i_2)$ is $\infty$. Otherwise, we run a close variant of the algorithm described in Section~\ref{sec:lcp}; note that $\ell_\omega$ can be considerably larger than one of $s_1$ or $s_2$. For Step 1, we define $n'=n=|s_1s_2|$; the other formulas do not change. We run the $\equal_\omega$ computations on $s_1$ and $s_2$ using Eq.~\eqref{eq:powers} to compute the fingerprints.
We extract the substrings of length $\ell'$ in Step 3 (analogously, $\ell''$ in Step 1) using the $\extract$ for circular strings, but do so only if $\ell' \le |s_1|$ (resp., $\ell' \le |s_2|)$; otherwise we keep accessing the original string using Eq.~\eqref{eq:powers}. 
The total amortized time to compute LCPs on omega extensions is thus
$\Oh(\log |s_1s_2|)$.

\subsection{Future work}
One feature that we would like to add to our data structure is allowing identification of conjugates. The rationale behind this is that a circular string can be represented by any of its linearizations, so these should all be regarded as equivalent. Furthermore, when the collection contains several conjugates of the same string, then this may be just an artifact caused by the data acquisition process. 

This could be solved by replacing each circular string with its necklace representative, that is, the unique conjugate that is lexicographically minimal in the conjugacy class, before applying \makestring; this representative is computable in linear time in the string length~\cite{lothaire3}. However, updates can change the lexicographic relationship of the rotations, and thus the necklace representative of the conjugacy class. 
Recomputing the necklace rotation of $s$ after each update would add worst-case $\Oh(|s|)$ time to our running times, which is not acceptable. 
Computing the necklace rotation after an edit operation, or more in general, after any one of our update operations, is an interesting research question, which to the best of our knowledge has not yet been addressed. 

\end{document}